\documentclass{article} 
\usepackage{latex8}
\usepackage{charter}

\usepackage{mathtools}
\usepackage{amsmath}
\usepackage{amsthm}
\usepackage{amssymb}
\usepackage{stmaryrd}

\usepackage{graphicx}

\graphicspath{{./}{Images/}}

\theoremstyle{definition}

\newtheorem{theorem}{Theorem}

\newtheorem{definition}[theorem]{Definition}

\newtheorem{schedulabilityTest}{Schedulability Test}
\newtheorem{example}{Example}

\newcommand{\drop}[1]{}
\newcommand{\equals}{\stackrel{\mathrm{def}}{=}}

\begin{document}

\title{Scheduling of Hard Real-Time Multi-Thread Periodic Tasks}
\author{Irina Lupu\hfill Jo\a"el Goossens\\
PARTS Research Center\\
Universit\a'e libre de Bruxelles (U.L.B.)\\
CP 212, 50 av. F.D. Roosevelt\\
1050 Brussels, Belgium\\
\{joel.goossens,irinlupu\}@ulb.ac.be}

\maketitle
\thispagestyle{empty}

\begin{abstract}
In this paper we study the scheduling of parallel and real-time recurrent tasks. Firstly, we propose a new parallel task model which allows recurrent tasks to be composed of several threads, each thread requires a single processor for execution and can be scheduled simultaneously. Secondly, we  define several kinds of real-time schedulers that can be applied to our parallel task model. We distinguish between two scheduling classes: hierarchical schedulers and global thread schedulers. We present and prove correct an exact schedulability test for each class. Lastly, we also evaluate the performance of our scheduling paradigm in comparison with Gang scheduling by means of simulations.
\end{abstract}

\Section{Introduction}\label{intro}

In this research, we consider the preemptive scheduling of hard real-time tasks on identical multiprocessor platforms (see~\cite{baker2,Baker2005An-analysis-of-}). We deal with \emph{parallel} real-time tasks, the case where each task instance (process in the following) may be executed on different processors \emph{simultaneously}. More specifically, each process is composed of several independent \emph{threads}, each thread requires one processor to be executed, consequently the process can progress upon several processors simultaneously. In this research we study the \emph{schedulability} problem  of recurring multi-thread tasks.

Nowadays, the design of parallel/multi-thread programs is common thanks to parallel programming paradigms like Message Passing Interface (MPI~\cite{Gorlatch1998A-Generic-MPI-I,Lusk1999Using-MPI-:-por}) Even better, sequential programs can be parallelized using tools like OpenMP (see~\cite{355074} for details).

\paragraph{Related Work.}

The literature concerning the scheduling of hard real-time and parallel recurring tasks is relatively poor. We can only report few models of parallel tasks and few results (schedulers and schedulability/feasibility tests). \mbox{Manimaran} et al.\@ in~\cite{mani98} consider the \emph{non-preemptive} EDF scheduling of periodic tasks. \mbox{Han} et al.\@ in~\cite{Han2006Predictability-} considered the scheduling of a (finite) set of real-time jobs allowing job parallelism while we consider the scheduling of either an infinite set of jobs (actually \emph{processes} in our terminology) or equivalently a set of periodic tasks. In a seminal work we contributed to the \emph{feasibility} problem of parallel tasks. In~\cite{CCG06b} we provided a task model which integrates job parallelism. We proved that the time-complexity of the feasibility problem of these systems is linear relatively to the number of (sporadic) tasks. In~\cite{lakshmanan2010scheduling} \mbox{Lakshmanan} et al.\@ consider the fork-join task model where each task is an alternate sequence of sequential code and $m$ parallel segments. They provided a partitioning algorithm and a competitive analysis for EDF and the fork-join task model. Regarding the \emph{schedulability} of recurrent real-time tasks, and to the best of our knowledge, we can only report results about the \emph{Gang} scheduling, where the execution requirement of processes corresponds as a $C \times v$ \emph{rectangle}, with the interpretation that a process requires \emph{exactly} $v$ processors \emph{simultaneously} for a duration of $C$ time units. Kato et al.\@ (see~\cite{KatoGang} for details) considered the EDF Gang scheduling and provided a \emph{sufficient} schedulability condition. We studied Fixed Task Priority (FTP) Gang scheduling~\cite{Goossens2010Gang-FTP-schedu} and we provided an exact schedulability test for periodic tasks.
\paragraph{This Research.}
In this paper we introduce a more realistic parallel task model, i.e., we study the scheduling of periodic and (parallel) \emph{multi-thread} tasks. Our main contribution is \emph{exact} schedulability tests for Fixed Subprogram Priority (FSP) and (FTP,~FSP) schedulers. Additionally, we show that Gang scheduling and multi-thread scheduling are incomparable and we present an empirical study which shows that, in most of the cases, multi-thread scheduling dominates Gang scheduling. Last but not least, in the further work section, we consider a realistic extension of our task model which allows tasks to be a sequence of parallel phases (each one having its own number of threads) and we present a first negative result. 

\paragraph{Paper Organization.}
This paper is organized as follows. Section~\ref{sec:model} provides our formal task model and important definitions. In Section~\ref{sec:schedulers} we present a taxonomy of schedulers dedicated  to our parallel task model. We present exact schedulability tests in Section~\ref{sec:tests} for two classes of schedulers. In Section~\ref{sec:study} we present an empirical study which shows that in most of the cases thread scheduling dominates Gang scheduling. Lastly, in Section~\ref{sec:conclusion} we conclude and consider a realistic extension of our parallel task model and present a first negative result.

\Section{Formal model and definitions}
\label{sec:model}

In this work we consider periodic multi-thread tasks, each task $\tau_{i}$ is characterized by the tuple 
\begin{equation*}
 (O_i, \{q_i^1, q_i^2, \ldots, q_i^{v_i}\}, D_i, T_i),
 \end{equation*}
where
\begin{itemize}
\item $O_i$ is the arrival instant, i.e., the moment of the first activation of the task since the system initialization;
\item $\{q_i^1, q_i^2, \ldots, q_i^{v_i}\}$ is the set of the $v_i$ subprograms of $\tau_i$; at run-time these subprograms generate threads which can be executed simultaneously, i.e., we allow task parallelism;
\item each subprogram $q_i^j$ ($1 \leq j \leq v_i$) is characterized by an individual worst-case execution time $C_i^j$;
\item $T_i$ is the period, i.e., the \emph{exact} inter-arrival time between two successive activations of the task;
\item $D_i$ is the relative deadline, i.e., the time by which the current instance of the task has to finish execution relatively to its arrival.
\end{itemize}

Throughout this paper, all timing characteristics in our model are assumed to be non-negative integers, i.e., they are multiples of some irreducible time interval (for example the CPU clock cycle is the smallest indivisible CPU time unit).

\paragraph{Task/process \& subprogram/thread.} In this work we will distinguish between off-line and run-time entities. We consider that a task is defined off-line while its instance exists only at run-time under the denomination \emph{process}. In the same vein we consider that a subprogram is defined off-line while its instance exists only at run-time under the denomination \emph{thread}. Consequently the scheduler manages processes and/or threads. Meanwhile the process or thread priority can (or cannot) be based on the static task and subprogram characteristics. 

Each task system $\tau$ is composed of $n$ periodic and parallel such tasks: $\tau = \{\tau_1, \ldots, \tau_n\}$. The deadline of each task is less than or equal to the period: $\forall i, D_i \leq T_i$ (\emph{constrained deadline model}). As each task $\tau_i$ has its own offset $O_i$, the task systems considered are \emph{asynchronous}. 

Since tasks are periodic, their subprograms have a periodic behavior as well. The $j^{\operatorname{th}}$ thread of the $k^{\operatorname{th}}$ process of $\tau_i$ is characterized by the following parameters: an arrival time $a_k^j = O_i + (k-1) \cdot T_i$, an execution demand $e_k^j = C_i^j$ and an absolute deadline $d_k^j = a_k^j + D_i$. 

The $k^{\operatorname{th}}$ process of $\tau_i$ is characterized by an arrival time $a_k = O_i + (k-1) \cdot T_i$, an execution demand expressed by the parallel execution demands of its threads $e_k = (e_k^1, \ldots, e_k^{v_i})$ and an absolute deadline $d_k = a_k + D_i$.

A task is said to be \emph{active} if it has a process with unfinished execution demand.

A task is characterized by the measure called \emph{utilization}: $u_i = \frac{\sum_{j=1}^{v_i}C_i^j}{T_i}$. This measure represents the portion of the platform capacity ($u_i \leq m$) requested by the task $\tau_i$ when executing. We also denote by $U \equals \sum_{i=1}^{n} u_{i}$ the total system utilization.  In the following, $P$ denotes the \emph{least common multiple} of all the periods in the task system $\tau$: $P = \operatorname{lcm}\{T_1, \ldots, T_n\}$.

The considered multiprocessor platform contains $m$ unit-capacity processors.

\begin{figure}
	\centering
   \includegraphics[width=6cm]{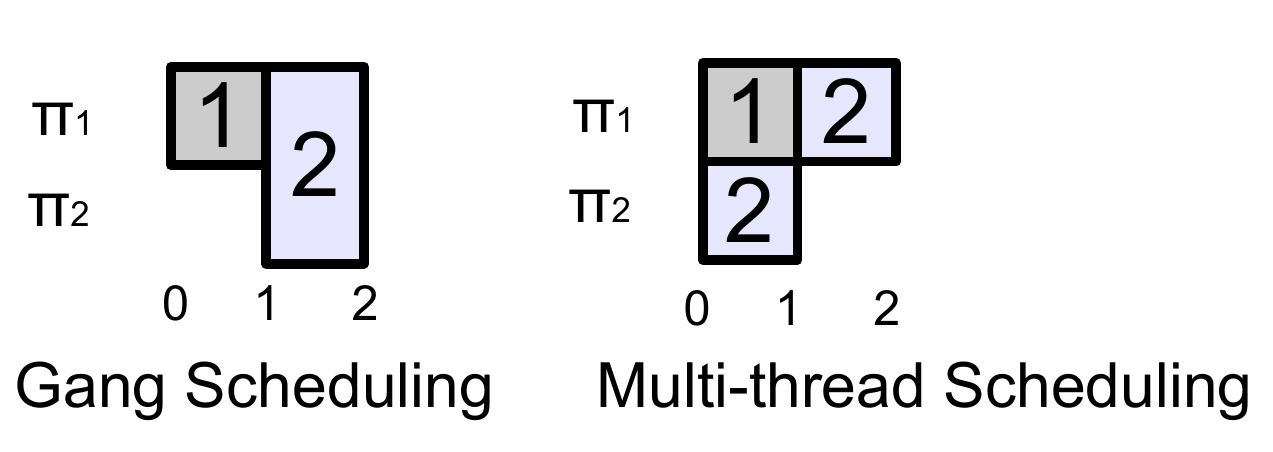}
   \caption{\label{fig:gang-thread}Gang scheduling vs.\@ thread scheduling}
\end{figure}

\paragraph{Gang scheduling versus Thread Scheduling.}
Figure~\ref{fig:gang-thread} illustrates a Gang and a thread scheduling for the ``same'' task set: $\tau_{1} = (0,\{1\},T_{1},D_{1}) > \tau_{2} = (0,\{1,1\},T_{2},D_{2})$. Focusing on $\tau_{2}$, Gang  scheduling has to manage the rectangle $C \times v = 1 \times 2$ while thread scheduling has to manage two 1-unit length threads.
 
From our point of view thread scheduling paradigm resolves the following Gang scheduling drawbacks:
\begin{enumerate}
	\item As exhibited in~\cite{Goossens2010Gang-FTP-schedu} Gang scheduling suffers from \emph{priority inversion}, i.e., a lower priority task can progress while an higher priority active task cannot.
	\item The number of processors required by a task must be not larger than the platform size. 
	\item Because of the requirement that the task must execute on exactly upon $v$ processors simultaneously very often many processors may be left idle while there is active tasks.
	\item As shown in~\cite{Goossens2010Gang-FTP-schedu} Gang FTP schedulers are not \emph{predictable}. Multi-thread FTP schedulers are proven predictable in Section~\ref{sec:tests}.
\end{enumerate}

In this proposal, the exact schedulability tests are based on \emph{feasibility intervals} with the following definition. 

\begin{definition}[Feasibility interval]
For any task system $\tau = \{\tau_1, \ldots, \tau_n\}$ and any $m$-unit capacity multiprocessor platform, the \emph{feasibility interval} is a finite interval such that if no deadline is missed while considering only the processes in this interval no deadline will ever be missed.
\end{definition}

In this paper we consider that the scheduling is \emph{priority-driven}: the threads are assigned \emph{distinct} priority levels. According to these priority levels the scheduler decides at each time $t$ what will be executed on the multiprocessor platform at that time instant: the $m$ highest (if any) priority threads will be executed simultaneously on the given platform. The thread-processor assignment is \emph{univocally} determined by the following rule: ``\emph{higher the priority, lower the  processor index}''. If less than $m$ threads are active the processors with the higher indexes are left idle.  

We  consider \emph{preemptive} scheduling: a higher priority  thread can interrupt the executing lower priority thread.

\Section{Taxonomy of schedulers}\label{sec:schedulers}

In this work we consider two classes of real-time schedulers for our parallel task model: \emph{hierarchical schedulers} and \emph{global thread schedulers}.

\begin{itemize}
\item At top-level \emph{hierarchical schedulers} manage processes with a process-level scheduling rule and use a second (low-level) scheduling rule two manage threads \emph{within} each process. 
\item \emph{Global thread schedulers} assign priorities to threads regardless of the tasks/subprograms that generated them.
\end{itemize}

In order to define rigorously our Hierarchical and Global schedulers we have to introduce the following schedulers.

\begin{definition}[Fixed Task Priority (FTP)]
A fixed task priority scheduler assigns a fixed and distinct priority to each task before the execution of the system. At run-time each process priority corresponds to its task priority.  
\end{definition}

Among the FTP schedulers we can mention \emph{Rate-Monotonic} (RM)~\cite{Liu:1973:SAM:321738.321743} and \emph{Deadline-Monotonic} (DM)~\cite{Audsley91hardreal-time}. 

\begin{definition}[Fixed Process Priority (FPP)]
A fixed process priority scheduler assigns a fixed and distinct priority to processes upon arrival. Each process preserves the priority level during its entire execution.
\end{definition}

The \emph{Earliest Deadline First} (EDF)~\cite{Liu:1973:SAM:321738.321743} scheduler is an example of FPP scheduler.
 
\begin{definition}[Dynamic Process Priority (DPP)]
A dynamic process priority scheduler assigns, at each time $t$, priorities to the active processes according to their run-time characteristics. Consequently, during its execution, a process may have different priority levels.
\end{definition}

The \emph{Least Laxity First} (LLF) scheduler is a DPP scheduler since the laxity is a dynamic process characteristic (see~\cite{Dertouzos:89} for details). 

In the same vein, the following schedulers can be defined at thread level:

\begin{definition}[Fixed Subprogram Priority (FSP)]
A fixed subprogram priority scheduler assigns a fixed and distinct priority to each subprogram before the execution of the system. At run-time each thread priority corresponds to its subprogram priority. 
\end{definition}

An example of FSP scheduler is the \emph{Longest Subprogram First} scheduler. 

\begin{definition}[Fixed Thread Priority (FThP)]
A fixed thread priority scheduler assigns a fixed and distinct priority to threads upon arrival. Each thread preserves the priority level during its entire execution.
\end{definition}

As far as we know, no FThP scheduler can be defined based \emph{only} on the characteristics of the tasks in our model, at least schedulers \emph{different} than the class FSP.

\begin{definition}[Dynamic Thread Priority (DThP)]
A dynamic thread priority scheduler assigns, at time $t$, priorities to the existing threads according to their characteristics. During its execution, a thread may have different priority levels.
\end{definition}

An example of DThP is LLF applied at thread level.


\SubSection{Hierarchical schedulers}
Hierarchical schedulers  are built following the next two steps:
\begin{enumerate}
\item at process level, one of the following schedulers is chosen in order to assign priorities to process: FTP, FPP and DPP.
\item for assigning priorities \emph{within} process, one of the following schedulers will be chosen: FSP, FThP, DThP.
\end{enumerate}
In the following an hierarchical scheduler will be denoted by the couple $(\alpha, \beta)$, where $\alpha \in \{\operatorname{FTP}, \operatorname{FPP}, \operatorname{DPP}\}$ and $\beta \in \{\operatorname{FSP}, \operatorname{FThP}, \operatorname{DThP}\}$.

\subsection{Global thread schedulers}
As global thread schedulers, the  FSP, FThP and DThP schedulers can be applied to a set of subprograms or threads regardless of the task that they belong to. 

Notice that some global thread schedulers are identical to some hierarchical ones. E.g., a total order between threads (i.e., a FThP scheduler) can ``mimic'' any hierarchical (FTP, FThP) scheduler.

\Section{Exact schedulability tests}\label{sec:tests}

In this section, we will present two exact schedulability tests: one for FSP and one for (FTP, FSP) schedulers. 

\SubSection{Exact FSP schedulability test}
The first step into defining the schedulability test for FSP schedulers is to prove that their schedules are periodic. The proof is based on the periodicity of FTP schedules when the FTP is applied to systems $\tau'$ with the following task model: a task $\tau_k' \in \tau'$ is characterized by $(O_k, C_k, D_k, T_k)$~\cite{CGG:11}. This model will be called in this paper the \emph{sequential task model}. The periodicity of FTP schedules for the sequential task model is stated in Theorem~\ref{theor1}. The tasks $\tau_k'$ in $\tau'$ ($1 \leq k \leq n$) are ordered by decreasing priority: $\tau_1' > \cdots > \tau_n'$.

\begin{theorem}[\cite{CGG:11}]
\label{theor1}
For any preemptive FTP scheduling algorithm A, if an asynchronous constrained deadline system $\tau' = \{\tau_1', \ldots, \tau_n'\}$ is A-feasible, then the A-schedule of $\tau'$ on $m$-unit capacity multiprocessor platform is periodic with a period of $P$ starting from instant $S_n$ where $S_i$ is defined as:
\begin{equation}
\begin{cases}
S_1 \equals O_1,  \\
S_i \equals \max\{O_i, O_i + \lceil \frac{S_{i-1} - O_i}{T_i} \rceil \cdot T_i \}, \\
\hspace*{1.5cm} \forall i \in \{2, 3, \ldots, n\}.
\end{cases} \label{eqn1}
\end{equation}
(Assuming that the execution times of each task is constant.)
\end{theorem}

In the following, we consider a task system $\tau$ with $ r \equals \sum_{i=1}^n v_i$ subprograms; each subprogram $q_i^j$ of the task $\tau_i$ is characterized by $(O_i, C_i^j, T_i, D_i)$. A FSP scheduler is used to assign priorities to the $r$ subprograms.  As these priorities are assigned regardless of the tasks to which the subprograms belong to, we can consider a simpler notation: each subprogram $q^{\ell}$ is characterized by $(O_{\ell}, C_{\ell}, T_{\ell}, D_{\ell})$ ($1 \leq {\ell} \leq r$); if $q^{\ell}$ corresponds to the $j^{\operatorname{th}}$ subprogram of $\tau_i$, then $O_{\ell} = O_i$, $C_{\ell} = C_i^j$, $D_{\ell} = D_i$, $T_{\ell} = T_i$. In the following we assume without loss of generality that subprograms are ordered by FSP decreasing priority: $q^1> \cdots > q^r$.

\begin{theorem}
For any preemptive FSP scheduling algorithm A, if an asynchronous constrained deadline system $\tau$ containing $r$ subprograms (regardless of the tasks they belong to) is A-feasible, then the A-schedule of $\tau$ on $m$ unit-capacity multiprocessor platform is periodic with a period of $P$ starting from instant $S_r^{*}$, where $S_i^{*}$ is defined as follows:
\begin{equation}
\begin{cases}
S_1^{*} \equals O_1,  \\
S_j^{*} \equals \max\{O_j, O_j + \lceil \frac{S_{j-1}^{*} - O_j}{T_j} \rceil \cdot T_j \}, \\
\hspace*{1.5cm} \forall j \in \{2, 3, \ldots, r\}.
\end{cases} \label{eqn2}
\end{equation}
(Assuming that the execution times of each subprogram is constant.)
\label{th2}
\end{theorem}
\begin{proof}
When using FSP, a parallel task system $\tau$ with $n$ tasks and $r$ subprograms can be seen a sequential task system $\tau'$ which contains $r$ sequential periodic tasks $\tau' = \{ \tau_1', \ldots, \tau_r'\}$ such that the task $\tau_{\ell}'$ ($1 \leq \ell \leq r$) has the same characteristics as the ${\ell}^{\operatorname{th}}$ subprogram of $\tau$ ($\tau_{\ell}':~(O_{\ell}, C_{\ell}, T_{\ell}, D_{\ell})$). From the FSP priority assignment on $\tau$, a FTP priority assignment for $\tau'$ can be defined: if $q^1> \cdots > q^r$ according to FSP, the corresponding sequential tasks have the following order $\tau_1' > \cdots > \tau_r'$ according to FTP.

By Theorem~\ref{theor1}, we know that the schedule of FTP on $\tau'$ is periodic with a period of $P$ starting with $S_r$. We can observe that $S_r$ has the same value as $S_r^{*}$. This means that the FSP schedule on $\tau$ is periodic with a period of $P$ starting with $S_r^{*}$.
\end{proof}

Theorem~\ref{th2} considers that execution times $C_{\ell}$ of a subprogram $q^{\ell}$ ($1 \leq \ell \leq r$) are \emph{constant}. In order to define the schedulability test for the FSP schedulers, we have to prove that they are \emph{predictable}.

\begin{definition}[Predictability]
Lets consider the sets of threads $J$ and $J'$ which differ only with regards to their the execution times: the threads in $J$ have executions times less than or equal to the execution times of the corresponding threads in $J'$. A scheduling algorithm $A$ is \emph{predictable} if, when applied independently on $J$ and $J'$, a thread in $J$ finishes execution before or at the same time as the corresponding thread in $J'$.
\end{definition} 

\begin{theorem}
FSP schedulers are \emph{predictable}. 
\label{th3}
\end{theorem}
\begin{proof}
We mentioned in the proof of the Theorem~\ref{th2} that the task system $\tau$ containing $r$ subprograms can be seen as a system $\tau'$ of $r$ sequential tasks such that a task $\tau_{\ell}'$ inherits the characteristics of the corresponding subprogram $q^{\ell}$ of $\tau$ ($1 \leq \ell \leq r$). A FTP priority assignment for $\tau'$ can be built following the priorities assigned by FSP to the corresponding subprograms in $\tau$: $q^1> \cdots > q^r$ gives $\tau_1' > \cdots > \tau_r'$.

In \cite{DBLP:conf/icdcs/HaL94} it is proven that for systems like $\tau'$, FTP schedulers are predictable on $m$ unit-capacity multiprocessor platforms. Since $\tau$ is equivalent to $\tau'$ and the FTP scheduler assigns the same priorities to sequential tasks as FSP to the corresponding subprograms, FSP schedulers are also predictable.
\end{proof}

Based on Theorem~\ref{th2} and \ref{th3}, we can define an exact feasibility test for FSP schedulers.
\begin{schedulabilityTest}
For any preemptive FSP scheduler $A$ and for any $A$-feasible asynchronous constrained deadline system $\tau$ containing $r$ subprograms (regardless of the tasks they belong to) on a $m$ unit-capacity multiprocessor platform, $[0, S_r^{*} + P)$ is a feasibility interval, where $S_i^{*}$ is defined by equation~\ref{eqn2}.
\end{schedulabilityTest}
\begin{proof}
This is a direct consequence of Theorems~\ref{th2} and~\ref{th3}.
\end{proof}

\SubSection{Exact (FTP, FSP) schedulability test}
The first step in the definition of the exact schedulability test for the (FTP, FSP) schedulers is to prove the periodicity of the feasible schedules.

\begin{theorem}
For any preemptive (FTP, FSP) scheduling algorithm A, if an asynchronous constrained deadline parallel task system $\tau = \{\tau_1, \ldots, \tau_n\}$  is A-feasible, then the A-schedule of $\tau$ on $m$ unit-capacity multiprocessor platform is periodic with a period of $P$ starting from instant $S_n$, where $S_i$ is defined by equation~\ref{eqn1} and tasks are ordered by decreasing priority: $\tau_1 > \tau_2 > \cdots > \tau_n$.
\label{th5}
\end{theorem}
\begin{proof}
Lets consider that the tasks in $\tau$ and their subprograms are ordered by decreasing priority: 
\begin{equation}
\tau_1 > \tau_2 > \cdots > \tau_n \text{   with   } q_i^1 > \cdots > q_i^{v_i}, \, \, \forall 1 \leq i \leq n. \nonumber
\end{equation}
Following these priority orders, we can define a FSP scheduler $B$ which assigns the following priorities to the $r = \sum_{i=1}^n v_i$ subprograms of $\tau$: 
\begin{eqnarray}
q_1^1 >  q_1^2 > \cdots > q_1^{v_1} > q_2^1 > q_2^2 > \cdots > q_2^{v_2} > \cdots \nonumber \\
\cdots > q_{n-1}^1 > \cdots > q_{n-1}^{v_{n-1}}  > q_n^1 > \cdots > q_n^{v_n}. 
\label{eqn5}
\end{eqnarray}

The FSP schedulers assign priorities to subprograms regardless of the tasks they belong to. So we can rewrite equation~\ref{eqn5} regardless of the tasks $\tau_1, \ldots, \tau_n$:
\begin{equation}
q^1 > \cdots > q^{v_1} > \cdots > q^{1+\sum_{i=1}^{n-1} {v_i} }> \cdots > q^r. \nonumber
\end{equation}

By Theorem~\ref{th2}, the schedule generated by $B$ is periodic with a period of $P$ from $S_r^{*}$. We can observe that the $S_j^{*}$ quantity defined by equation~\ref{eqn2} represents the instant of the first arrival of $q^j$ at or after time instant $S_{j-1}^{*}$. Since all the subprograms belonging to a task $\tau_i$ ($1 \leq i \leq n$) have the same activation times and the same periods and $B$ assigns consecutive priorities to the subprograms of the same task (as seen in equation~\ref{eqn5}):
\begin{equation}
 S_j^{*} = S_{j-1}^{*}, \forall j: \text{ } (1+ \sum_{k=1}^{i-1}{v_k})  \leq j \leq \sum_{k=1}^{i}{v_k}
 \label{eqn6}
\end{equation}
 and $1 \leq i \leq n$.

Furthermore, we can observe that $S_1^{*} = S_1 = O_1$. From this fact and equation~\ref{eqn6}, we can conclude:
\begin{eqnarray*}
 S_1^{*} = S_1 \Rightarrow \\
 S_{v_1 +1}^{*} = S_2 \Rightarrow \\
 \vdots \\
 S_{1+\sum_{i=1}^{n-1} v_i}^{*} = S_n.
 \end{eqnarray*}
 \label{eqnARR}
The $B$-schedule is then periodic with a period of $P$ starting from $S_n$. Since the $B$-schedule is the same as the one generated by $A$, the $A$-schedule is also periodic with a period of $P$ starting from $S_n$.
\end{proof}

We will now prove that the (FTP, FSP) schedulers are also predictable.

\begin{theorem}
(FTP, FSP) schedulers are \emph{predictable} .
\label{th6}
\end{theorem}
\begin{proof}
Since based on any (FTP, FSP) scheduler we can define a FSP scheduler as shown in the proof of the Theorem~\ref{th5} and since, by Theorem~\ref{th3}, FSP schedulers are predictable, (FTP, FSP) schedulers are predictable as well.
\end{proof}

We will now define the exact schedulability test for (FTP, FSP) schedulers. 
\begin{schedulabilityTest}
For any preemptive (FTP, FSP) scheduler $A$ and for any $A$-feasible asynchronous constrained deadline parallel task system $\tau = \{\tau_1, \ldots, \tau_n\}$ on a $m$ unit-capacity multiprocessor platform, $[0, S_n + P)$ is a feasibility interval, where $S_i$ is defined by equation~\ref{eqn1}.
\label{sched2}
\end{schedulabilityTest}
\begin{proof}
This is a direct consequence of Theorems~\ref{th5} and~\ref{th6}.
\end{proof}

%

\SubSection{Gang and thread scheduling are incomparable} 

In this section we will show that Gang FTP and thread hierarchical (FTP,~FSP) schedulers are incomparable ---in the sense that there are task systems which are schedulable using Gang scheduling approaches and not by thread scheduling approaches, and conversely.

The considered FTP scheduler is DM~\cite{Audsley91hardreal-time}: the priorities assigned to tasks by DM are inversely proportional to the relative deadlines. The FSP scheduler is called \emph{Index Monotonic} (IM) and it assigns priorities as follows: \emph{the lower the index of the subprogram within the task, the higher the priority}. 

In the following examples, the task offsets of the considered systems are equal to 0 and the feasible schedules are periodic from 0 with a period of $P$ (Theorem~\ref{th5} and~\cite{Goossens2010Gang-FTP-schedu}).

\begin{example} This first example presents a task system that is unschedulable by Gang DM, but schedulable by (DM,~IM) on a 2-unit capacity multiprocessor platform. The tasks in the system $\tau = \{ \tau_1, \tau_2, \tau_3\}$  have the following characteristics: $\tau_1: (0, \{2\}, 3, 3)$, $\tau_2: (0, \{3\}, 4, 4)$ and  $\tau_3: (0, \{2,2\}, 12, 12)$. According to DM $\tau_1 > \tau_2 > \tau_3$.
\begin{figure}[!h]
	\centering
   \includegraphics[width=6cm]{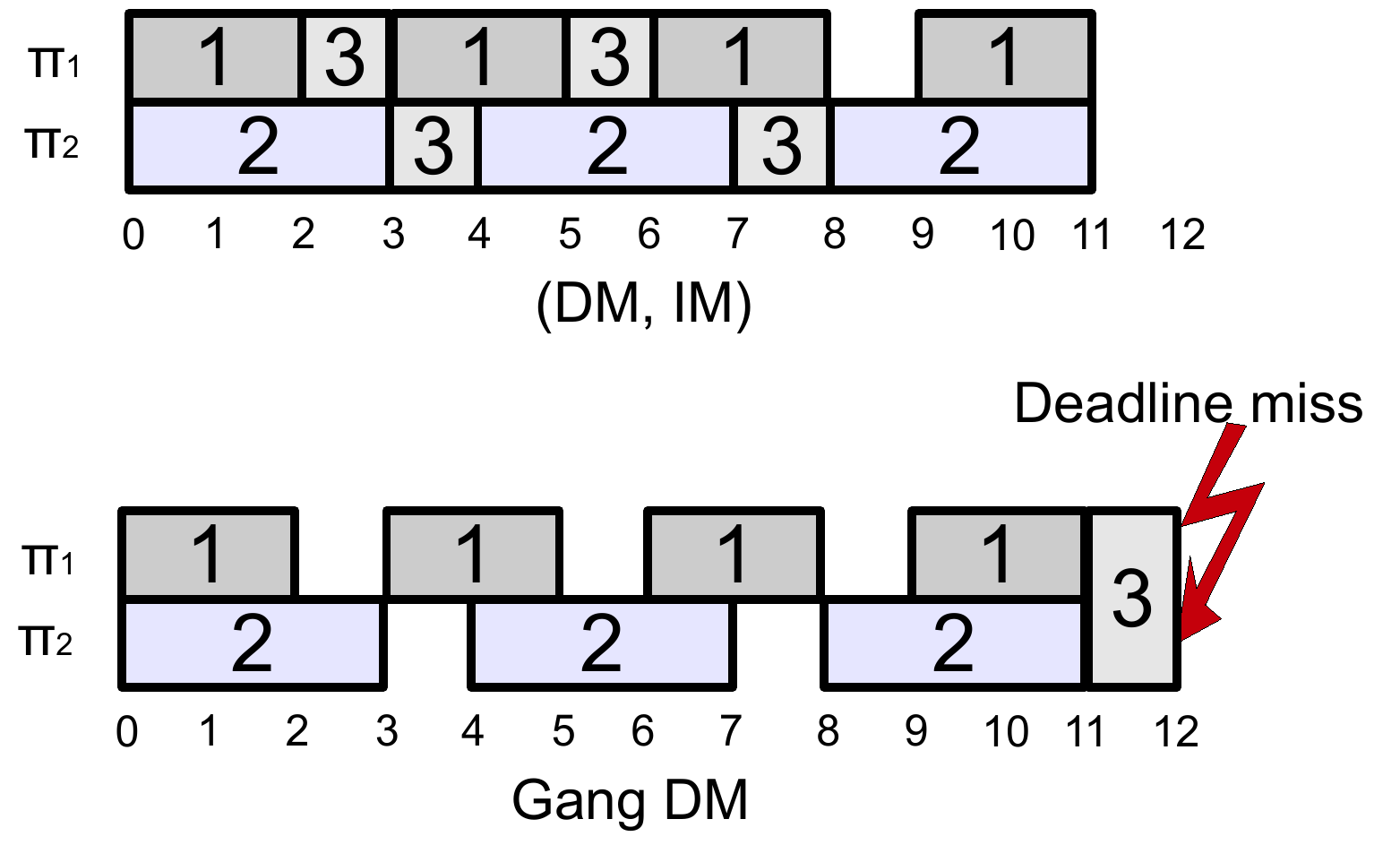}
   \caption{Gang DM unfeasible, (DM,~IM) feasible}
   \label{Exem1}
\end{figure}

We can observe in figure~\ref{Exem1} that according to Gang DM, task $\tau_3$ has to wait for 2 available processors simultaneously in order to execute. This is the case at time instant 11; though, at time 12 the task has unfinished execution demand and it misses its deadline.

In the case of (DM,~IM), $\tau_1$ and $\tau_2$ execute at the same moments and on the same processors as in Gang DM. The difference is that $\tau_3$ can start executing its first process at time instant 2 since one processor is available. Taking advantage of the fact that the processors are left idle by $\tau_1$ and $\tau_2$ at some moments in time, the first process of $\tau_3$ (which is the only $\tau_3$ process in the interval $[0, 12)$) finishes execution at time instant 8. No deadline is missed, therefore, the system is schedulable by (DM,~IM).
\end{example}

\begin{example} The second example presents a task system $\tau = \{\tau_1, \tau_2, \tau_3\}$ which is schedulable with Gang DM, but unschedulable with (DM,~IM) on a 3-unit capacity multiprocessor platform. The tasks of the system have the following characteristics: $\tau_1: (0, \{3,3\}, 4, 4)$,  $\tau_2: (0, \{1,1\}, 5, 5)$ and  $\tau_3: (0, \{9\}, 10, 10)$. According to DM $\tau_1 > \tau_2 > \tau_3$.

In figure~\ref{Exem2}, we can observe that according to Gang DM, at instant 0, $\tau_1$ is assigned to 2 of the 3 processors in the platform. Since there is only one processor left, $\tau_2$ cannot execute, therefore $\tau_3$ starts its execution on the third processor. At time instant 3, 2 processors are available and, consequently, $\tau_2$ may start executing, etc. No deadline is missed in the time interval $[0, 12)$, therefore the system is Gang DM feasible. 

According to (DM,~IM), even if $\tau_1$ occupies 2 processors of the 3 in the platform, $\tau_2$ may start executing on the third a first thread from time instant 0 to time instant 1. The second thread of its first process will execute on the third processor from time instant 1 to time instant 2. We can conclude that $\tau_3$ will miss its deadline at time instant 10 since it has 9 units of execution demand and only 6 time units available \emph{until} its deadline.
\begin{figure}[!h]
	\centering
   \includegraphics[width=8cm]{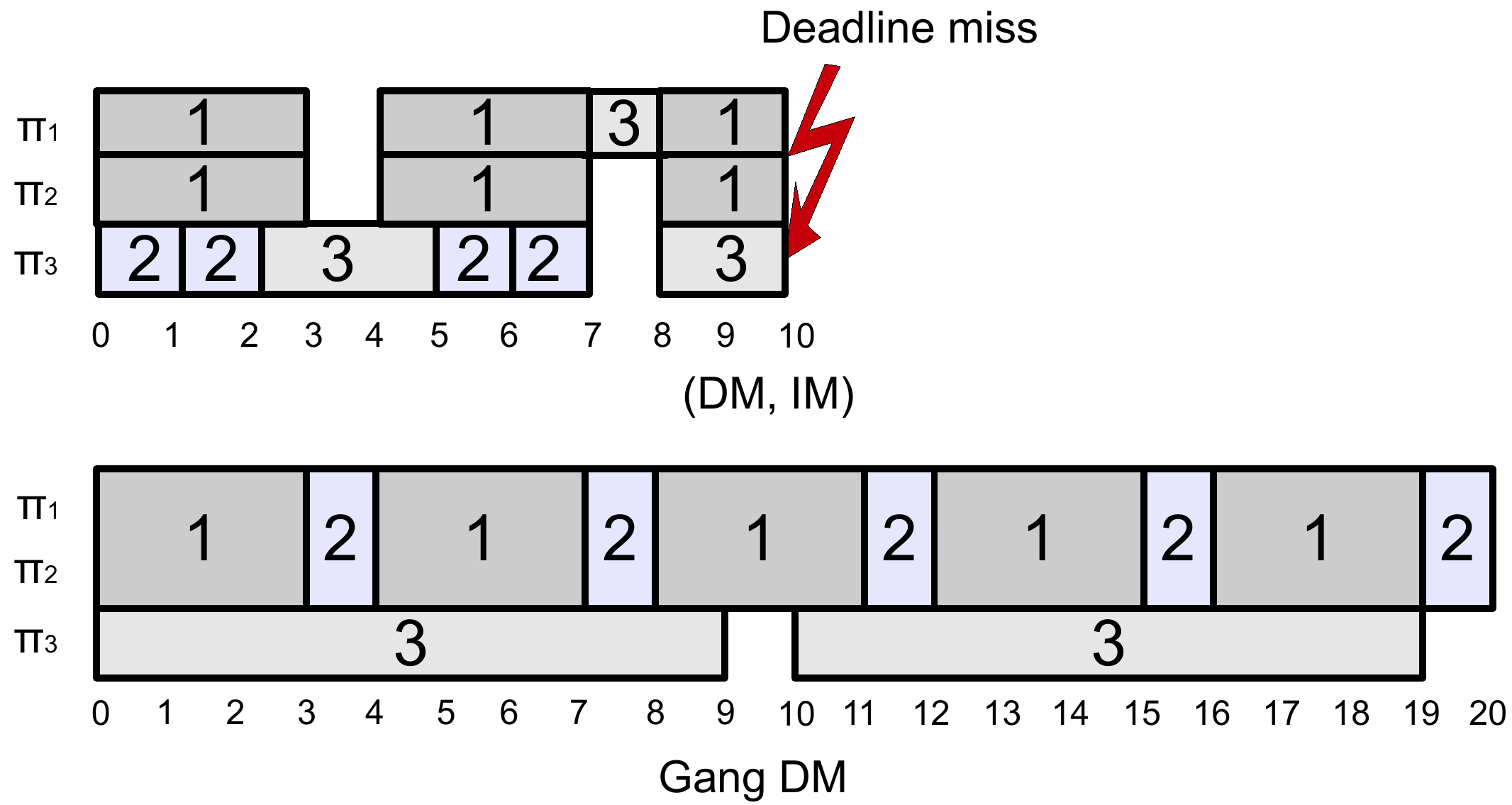}
   \caption{Gang DM feasible, (DM,~IM) unfeasible}
   \label{Exem2}
\end{figure}
\end{example}

\Section{Empirical Study}\label{sec:study}
The purpose of this empirical study is to evaluate the performance of multi-thread schedulers compared with the one of Gang schedulers. More specifically, the chosen multi-thread scheduler is the (FTP, FSP) scheduler (DM,~IM). Among the Gang schedulers, we consider Gang DM. 


Since Gang FTP schedulers are not predictable, in this study we consider constant execution times. From~\cite{Goossens2010Gang-FTP-schedu} and the Schedulability test~\ref{sched2}, we know that we have to simulate both Gang FTP and (FTP, FSP) schedulers in the time interval $[0, S_n + P)$ in order to conclude if the task system is feasible with one of them or with both.

In this empirical study, the execution times of all the subprograms of a task are considered equal. 
\SubSection{Evaluation criteria}
\label{subsect:evcrit}

Gang DM and (DM,IM) are evaluated according to the following criteria:
\begin{itemize}
\item success ratio;
\item worst-case response time of the lowest priority task in the system.
\end{itemize}

\begin{definition}[Response time]
The response time of a process represents the duration between the arrival of the process and the moment it finishes execution.
\end{definition}

\begin{definition}[Worst-Case Response Time (WCRT)]
The worst-case response time of a task is the maximum response time of its processes.
\end{definition}

\paragraph{Success ratio.}  The success ratio represents the portion of successfully scheduled task systems and it is defined as follows:
\begin{equation}
\frac{\operatorname{number \, of \, successfully \, scheduled \, task \, systems}}{\operatorname{total \, number \, of \, considered \, task \, systems}} \nonumber
\end{equation}

\paragraph{Worst-case response time.}  If system $\tau$ is schedulable, the worst-case response times of a task $\tau_i \in \tau$ for Gang DM and (DM,~IM) are calculated according to its processes within the time interval $[0, S_n + P)$. 

For each feasible task system with both Gang DM and (DM,~IM) we compare the WCRTs (computed for each of the two schedulers) of the lowest priority task. For a given system utilization, we count separately the tasks systems where the (DM,~IM) WCRT is \emph{strictly} inferior to the one computed under Gang DM and inversely. Consequently, the uncounted task systems are those where the computed WCRTs are equal for the two schedulers.


\SubSection{Task system generation methodology}
The procedure for task system generation is the following: individual tasks are generated and added to the system until the total system utilization exceeds the platform capacity ($m$).

The characteristics of a task $\tau_i$ are integers and they are generated as follows:
\begin{enumerate}
\item the period $T_i$ is uniformly chosen from $[1,250]$;
\item the offset is uniformly chosen from $[1, T_i]$;
\item the utilization $u_i$ of the task is inferior to $m$ and it is generated using the following distributions: 
\begin{itemize}
\item uniform distribution between $[\frac{1}{T_i}, m]$;
\item bimodal distribution: light tasks have an uniform distribution between $[\frac{1}{T_i}, \frac{m}{2}]$, heavy tasks have an uniform distribution between $[\frac{m}{2}, m]$; the probability of a task being heavy is of $\frac{1}{3}$;
\item exponential distribution of mean $\frac{m}{4}$;
\item exponential distribution of mean $\frac{m}{2}$;
\item exponential distribution of mean $\frac{3 \cdot m}{4}$.
\end{itemize}
\item $v_i$ is uniformly chosen from $[1,m]$;
\item since we consider that all the subprograms of a task $\tau_i$ have equal execution times, it is sufficient to compute a single execution time value: $C_i = \frac{u_i \cdot T_i}{v_i}$;
\item the deadline is uniformly chosen between $[C_i, T_i]$.
\end{enumerate}
We use several distributions (with different means) in order to generate a wide variety of task systems and, consequently, to have more accurate simulation results. 

The generated systems have a least common multiple of the task periods bounded by 5,000,000. During simulations, we considered multiprocessor platforms containing 2, 4, 8 and 16 processors. A total of 450,000 task systems were generated.
\SubSection{Results}
Firstly, we will analyze the success ratio of the two schedulers.

\paragraph{Success ratio.} Each of the following figures contains 3 plots: one represents the success ratio of the (DM,~IM) multi-thread scheduler, a second one the success ratio of Gang DM and a third one expresses the ratio of task systems successfully scheduled by both of them. 

The performance gap between the two schedulers in terms of success ratio is growing as the number of processors grows, as shown in figures~\ref{Graph1}--\ref{Graph4}. In the case of the 2 processors platform, they have similar performances. On 4 processors, the (DM,~IM) successfully schedules at most 10\% more task systems than Gang DM (this is attained at utilization 2.8). On 8 processors it schedules 12\% more task systems (at utilization 5.2) and on 16 processors it can schedule 14\% more task systems than Gang DM (at instant 10.4).

We can also observe in each of the figures~\ref{Graph1}--\ref{Graph4} that (DM,~IM) and Gang DM are incomparable since the corresponding plots are above the ``both'' plot in each of these figures. Moreover the amount of \emph{additional} systems that thread scheduling can manage is quite better. In figure~\ref{Graph1} we can see that on a 2 processors platform, the amount of additional systems that (DM,~IM) can schedule is at most 2 times higher (at utilization 1.6) than the one of additional systems scheduled by Gang DM. In the case of a 4 processors platform, (DM,~IM) can perform 4.3 times better than Gang DM (at utilization 2.8), on 8 processors 5.4 times (at utilization 5.2) and, finally, on 16 processors it can perform 7.5 times better than Gang DM (at utilization 10.4). 

\begin{figure}[!h]
	\centering
   \includegraphics[width=8cm]{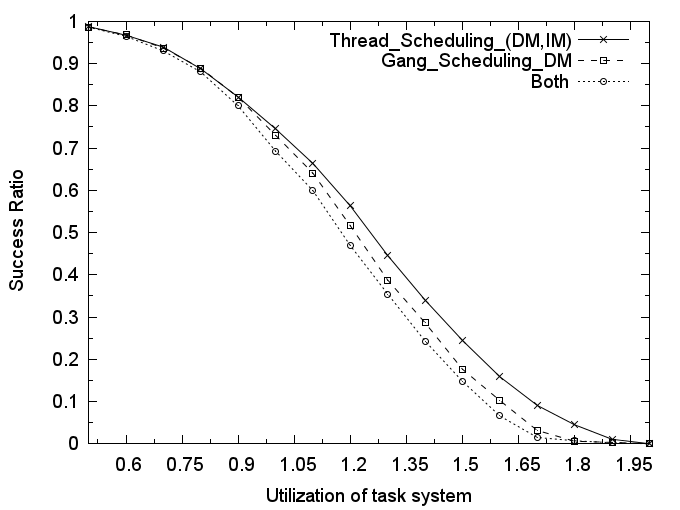}
   \caption{Success Ratio: 2 processors}
   \label{Graph1}
\end{figure}

\begin{figure}[!h]
	\centering
   \includegraphics[width=8cm]{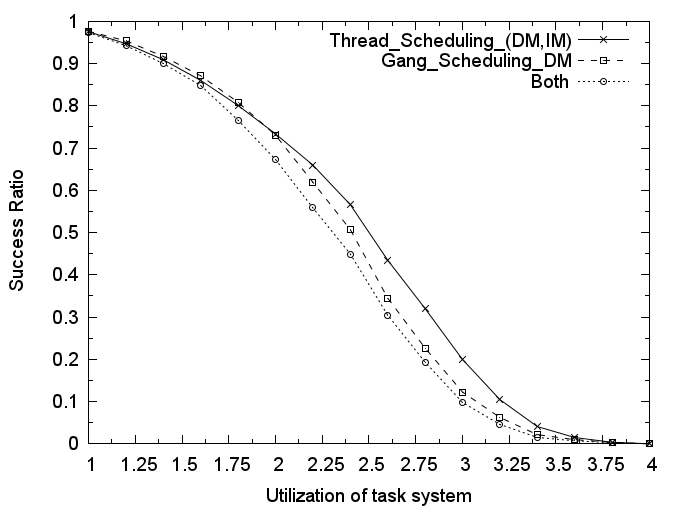}
   \caption{Success Ratio: 4 processors}
   \label{Graph2}
\end{figure}

\begin{figure}[!h]
	\centering
   \includegraphics[width=8cm]{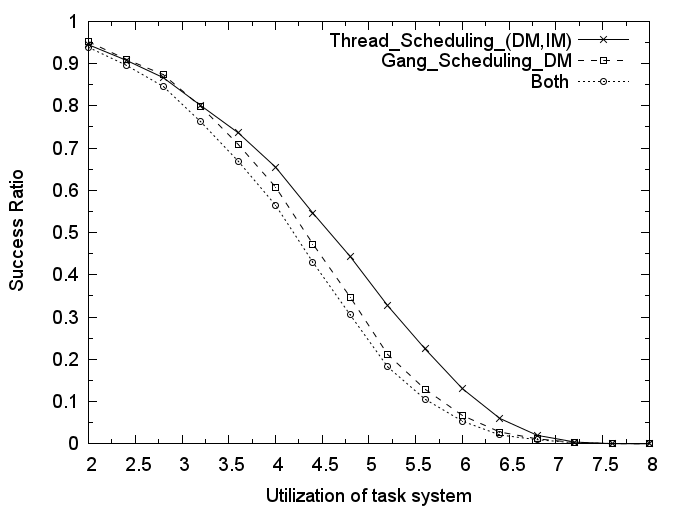}
   \caption{Success Ratio: 8 processors}
   \label{Graph3}
\end{figure}

\begin{figure}[!h]
	\centering
   \includegraphics[width=8cm]{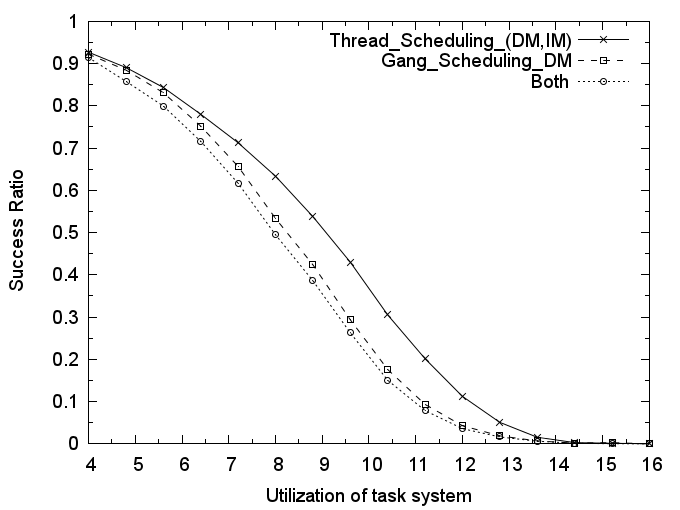}
   \caption{Success Ratio: 16 processors}
   \label{Graph4}
\end{figure}

\paragraph{Response time.} In the following we will reference the figures~\ref{Graph5}--\ref{Graph7}. The utilization of the considered systems is superior to 25\% and inferior to 90\% of the platform capacity.

In each figure, there are two plots: one that marks the portion of task systems (among those which are schedulable by both Gang DM and (DM,~IM)) where the (DM,~IM) WCRT of the lowest priority task is \emph{strictly} inferior to the one computed under Gang DM; a second plot marks the contrary behavior.

The results of our simulations showed that, for systems executed on 2 processors, the WCRTs of the lowest priority task under the two schedulers are equal. 

On 4, 8 and 16-unit capacity multiprocessor platforms, (DM,~IM) outperforms Gang DM as we can see in figures~\ref{Graph5}--\ref{Graph7}. For all of the three cases, the (DM,~IM) WCRT is inferior to the Gang DM WCRT in at most 50\% of the considered task systems. The lowest performance gap is observed on the 4 processors platform (at utilization 1.4) and, even in that case, (DM,~IM) performs 8\% better than Gang DM.

\begin{figure}[!h]
	\centering
   \includegraphics[width=8cm]{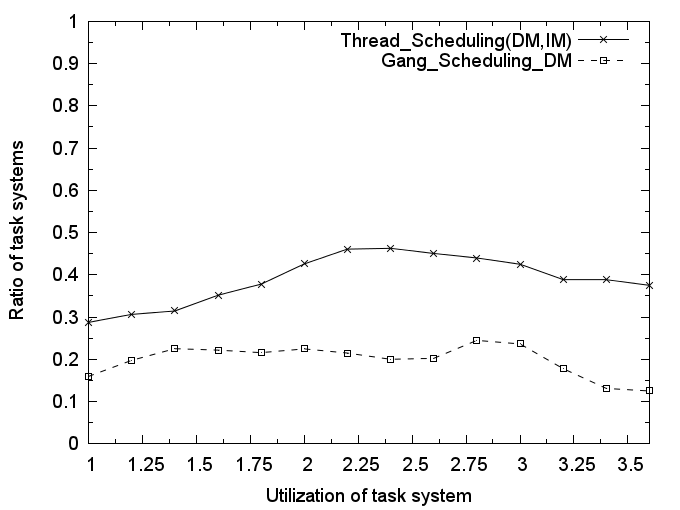}
   \caption{WCRT comparison: 4 processors}
   \label{Graph5}
\end{figure}

\begin{figure}[!h]
	\centering
   \includegraphics[width=8cm]{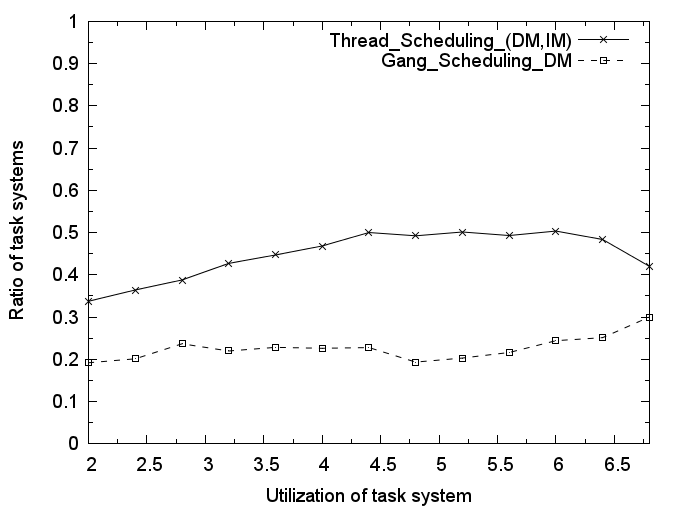}
   \caption{WCRT comparison: 8 processors}
   \label{Graph6}
\end{figure}

\begin{figure}[!h]
	\centering
   \includegraphics[width=8cm]{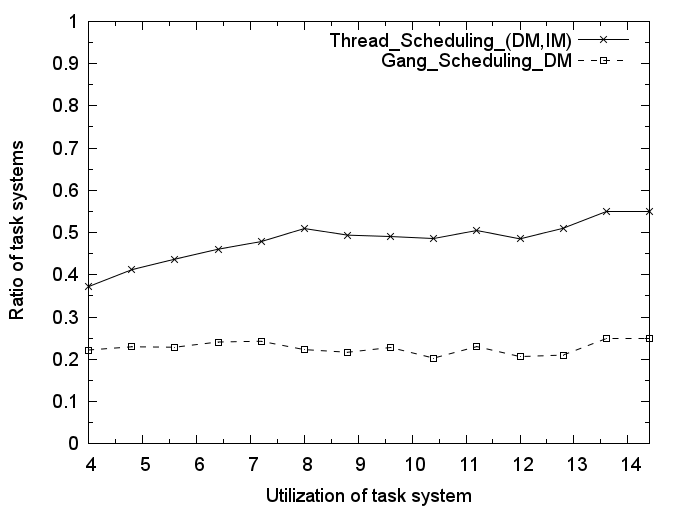}
   \caption{WCRT comparison: 16 processors}
   \label{Graph7}
\end{figure}

\Section{Conclusions and future work}\label{sec:conclusion}

In this work we considered the multi-thread scheduling for parallel real-time systems. The main advantage of this model is that it does not require \emph{all} threads of a same task to execute simultaneously as Gang scheduling does. 

We defined in this paper the several types of priority-driven schedulers dedicated  to our parallel task model and scheduling method. We distinguished between hierarchical schedulers (that firstly assign distinct priorities at system level and secondly, within each task) and global thread schedulers (that do not take into account the original tasks when priorities are assigned at thread level).

We showed that, contrary to Gang FTP, the hierarchical and global thread schedulers based on FTP and FSP are predictable. Based on this property and the periodicity of their schedules, we defined two exact schedulability tests.

Even though the Gang and multi-thread schedulers are, as we have shown, incomparable, the empirical study confirmed the intuition that our approach outperforms Gang scheduling. In terms of success ratio, the performance gap increases as the number of processors grows.

\paragraph{Future work.}
While we have shown the benefits of our thread scheduling model in comparison with Gang scheduling, in further work we would like to extend our model to have more realistic parallel task models. In particular we aim to consider \emph{multi-phase} tasks, in the sense that a task can have an arbitrary number of phases that have to be executed \emph{sequentially}. Each phase would be characterized by its own degree of parallelism. E.g., the task could be composed of initialization, computing and finalization phases; solely the computing phase can be parallelized, others are sequential. Unfortunately, we have to report a first negative result, i.e., multi-phase multi-thread hierarchical schedulers are not predictable.

We extend the model as follows, each task  $\tau_{i}$ is characterized by: 
\begin{equation*}
 (O_i, \{\phi_i^1, \phi_i^2, \ldots, \phi_i^{\ell_{i}}\}, D_i, T_i)
\end{equation*}
with the interpretation that the task is composed by a sequence of $\ell_{i}$ phases, each phase $\phi_{i}^{j}$ is characterized by a set of subprograms $q_{i,j}^1, q_{i,j}^2, \ldots, q_{i,j}^{v_{i,j}}$ ($v_{i,j}$ is the number of subprograms in the phase $j$ of task $i$). Each thread is defined as our original model.

Consider for instance $\tau_{1} = (0, \phi_{1}^{1}=\{2\}, \phi_{1}^{2}=\{2,2,2\}) > \tau_{2} = (1, \phi_{2}^{1}=\{1\})$ (we omit here deadline and period characteristics), Figure~\ref{fig:unpredictable.pdf} shows that if we consider the worst-case duration of each thread of $\tau_{1}$ the single thread of $\tau_{2}$ complete at time 2 while if we reduce the duration of the first phase of task $\tau_{1}$ by one time unit then $\tau_{2}$ completes at time 4. The example shows the non predictability of multi-phase multi-thread hierarchical schedulers. 

\begin{figure}[!h]
	\centering
   \includegraphics[width=4.5cm]{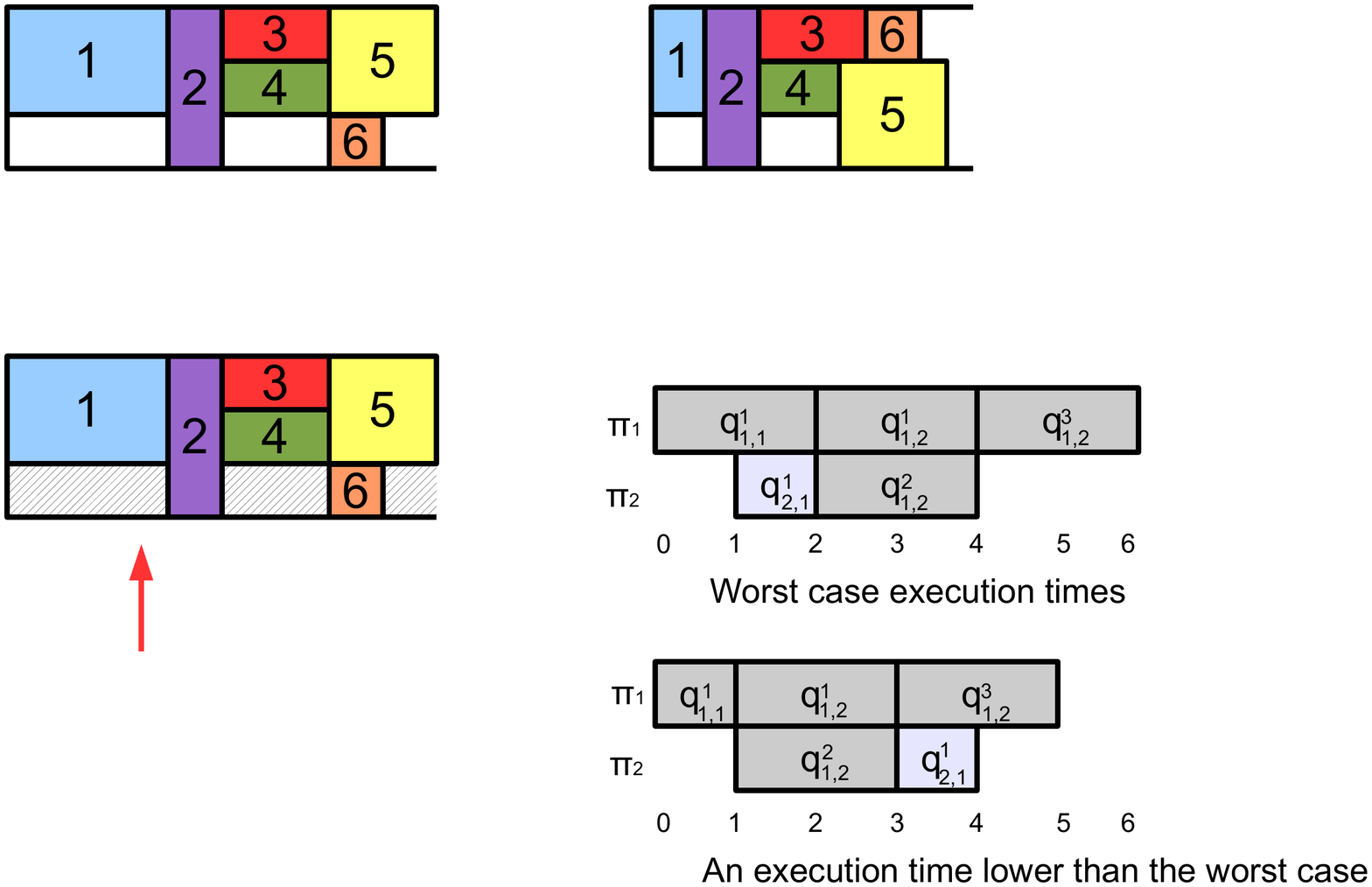}
   \caption{Multiphase : unpredictability}
   \label{fig:unpredictable.pdf}
\end{figure}



\bibliographystyle{latex8}
\bibliography{biblio}

\end{document}